\documentclass{llncs}

\usepackage{hyperref}

\usepackage[english]{babel}
\usepackage[utf8]{inputenc}
\usepackage[T1]{fontenc}

\usepackage{amsmath}
\usepackage{amsfonts}
\usepackage{amssymb}
\usepackage{stmaryrd}

\usepackage{url}
\usepackage{array}
\usepackage{multirow}

\usepackage{comment}

\usepackage{prettyref}
\newrefformat{def}{Def.~\ref{#1}}
\newrefformat{ex}{Example~\ref{#1}}
\newrefformat{fig}{Fig.~\ref{#1}}
\newrefformat{pro}{Property~\ref{#1}}
\newrefformat{pps}{Proposition~\ref{#1}}
\newrefformat{lem}{Lemma~\ref{#1}}
\newrefformat{thm}{Theorem~\ref{#1}}
\newrefformat{sec}{Sect.~\ref{#1}}
\newrefformat{ssec}{Subsect.~\ref{#1}}
\newrefformat{suppl}{Appendix~\ref{#1}}
\newrefformat{alg}{Algorithm~\ref{#1}}
\newrefformat{tab}{Table~\ref{#1}}
\newrefformat{eq}{Eq.~\eqref{#1}}
\def\pref{\prettyref}

\newcommand{\ie}{\textit{i.e.} }

\usepackage{algorithm,algpseudocode}

\def\precond#1{ {}^\bullet #1 }
\def\postcond#1{ #1 {}^\bullet}

\def\DEF{\stackrel{\Delta}=}
\def\EQDEF{\stackrel{\Delta}\Leftrightarrow}
\def\f#1{\mathsf{#1}}

\def\powerset{\wp}

\def\stimes{\tilde{\times}}
\def\sprod{\widetilde{\prod}}

\def\disabling{\ominus}

\def\obj#1#2{{#1\!\rightarrow^*\!#2}}
\def\Obj{\mathbf{Obj}}
\def\Proc{\mathbf{LS}}
\def\Sol{\mathbf{Sol}}

\def\sol{\f{sol}}

\def\A{\mathcal{A}}
\def\cwA{\A_\ctx^\w}
\def\cwNodes{V_\ctx^\w}
\def\cwEdges{E_\ctx^\w}

\def\fsimplN{\zeta^N}
\def\simplN#1{\fsimplN( #1)}
\def\card#1{\##1}

\def\w{\omega}
\def\ctx{\varsigma}

\def\val{\mathbb V}
\def\Val{\mathbf{Val}}
\def\update{\f{update}}
\def\M{\mathcal M}
\def\rank{\f{rank}}

\def\childs{\f{children}}
\def\parents{\f{parents}}

\def\SetNSets#1{\powerset(\powerset^{\leq N}(#1))}

\def\Obs{\mathcal Obs}
\def\Sys{\mathcal Sys}

\def\floodname{$\cwA$-Minimal-Cut-NSets}

\def\L{\mathcal L}

\usepackage{tikz}
\usetikzlibrary{arrows,shapes}
\usetikzlibrary{positioning} 
\usetikzlibrary{matrix,decorations.pathmorphing,shapes.geometric}
\usetikzlibrary{petri}

\usepackage{ifthen}
\usepackage{tikz}
\usetikzlibrary{arrows,shapes}

\definecolor{lightgray}{rgb}{0.8,0.8,0.8}
\definecolor{lightgrey}{rgb}{0.8,0.8,0.8}

\tikzstyle{boxed ph}=[]
\tikzstyle{sort}=[fill=lightgray,rounded corners]
\tikzstyle{process}=[circle,draw,minimum size=15pt,fill=white,
font=\footnotesize,inner sep=1pt]
\tikzstyle{black process}=[process, fill=black,text=white, font=\bfseries]
\tikzstyle{gray process}=[process, draw=black, fill=lightgray]
\tikzstyle{current process}=[process, draw=black, fill=lightgray]
\tikzstyle{process box}=[white,draw=black,rounded corners]
\tikzstyle{tick label}=[font=\footnotesize]
\tikzstyle{tick}=[black,-]
\tikzstyle{hit}=[->,>=angle 45]
\tikzstyle{selfhit}=[min distance=30pt,curve to]
\tikzstyle{bounce}=[densely dotted,>=stealth',->]
\tikzstyle{hl}=[font=\bfseries,very thick]
\tikzstyle{hl2}=[hl]
\tikzstyle{nohl}=[font=\normalfont,thin]

\newcommand{\currentScope}{}
\newcommand{\currentSort}{}
\newcommand{\currentSortLabel}{}
\newcommand{\currentAlign}{}
\newcommand{\currentSize}{}

\newcounter{la}
\newcommand{\TSetSortLabel}[2]{
  \expandafter\repcommand\expandafter{\csname TUserSort@#1\endcsname}{#2}
}
\newcommand{\TSort}[4]{
  \renewcommand{\currentScope}{#1}
  \renewcommand{\currentSort}{#2}
  \renewcommand{\currentSize}{#3}
  \renewcommand{\currentAlign}{#4}
  \ifcsname TUserSort@\currentSort\endcsname
    \renewcommand{\currentSortLabel}{\csname TUserSort@\currentSort\endcsname}
  \else
    \renewcommand{\currentSortLabel}{\currentSort}
  \fi
  \begin{scope}[shift={\currentScope}]
  \ifthenelse{\equal{\currentAlign}{l}}{
    \filldraw[process box] (-0.5,-0.5) rectangle (0.5,\currentSize-0.5);
    \node[sort] at (-0.2,\currentSize-0.4) {\currentSortLabel};
   }{\ifthenelse{\equal{\currentAlign}{r}}{
     \filldraw[process box] (-0.5,-0.5) rectangle (0.5,\currentSize-0.5);
     \node[sort] at (0.2,\currentSize-0.4) {\currentSortLabel};
   }{
    \filldraw[process box] (-0.5,-0.5) rectangle (\currentSize-0.5,0.5);
    \ifthenelse{\equal{\currentAlign}{t}}{
      \node[sort,anchor=east] at (-0.3,0.2) {\currentSortLabel};
    }{
      \node[sort] at (-0.6,-0.2) {\currentSortLabel};
    }
   }}
  \setcounter{la}{\currentSize}
  \addtocounter{la}{-1}
  \foreach \i in {0,...,\value{la}} {
    \TProc{\i}
  }
  \end{scope}
}

\newcommand{\TTickProc}[2]{ 
  \ifthenelse{\equal{\currentAlign}{l}}{
    \draw[tick] (-0.6,#1) -- (-0.4,#1);
    \node[tick label, anchor=east] at (-0.55,#1) {#2};
   }{\ifthenelse{\equal{\currentAlign}{r}}{
    \draw[tick] (0.6,#1) -- (0.4,#1);
    \node[tick label, anchor=west] at (0.55,#1) {#2};
   }{
    \ifthenelse{\equal{\currentAlign}{t}}{
      \draw[tick] (#1,0.6) -- (#1,0.4);
      \node[tick label, anchor=south] at (#1,0.55) {#2};
    }{
      \draw[tick] (#1,-0.6) -- (#1,-0.4);
      \node[tick label, anchor=north] at (#1,-0.55) {#2};
    }
   }}
}
\newcommand{\TSetTick}[3]{
  \expandafter\repcommand\expandafter{\csname TUserTick@#1_#2\endcsname}{#3}
}

\newcommand{\myProc}[3]{
  \ifcsname TUserTick@\currentSort_#1\endcsname
    \TTickProc{#1}{\csname TUserTick@\currentSort_#1\endcsname}
  \else
    \TTickProc{#1}{#1}
  \fi
  \ifthenelse{\equal{\currentAlign}{l}\or\equal{\currentAlign}{r}}{
    \node[#2] (\currentSort_#1) at (0,#1) {#3};
  }{
    \node[#2] (\currentSort_#1) at (#1,0) {#3};
  }
}
\newcommand{\TSetProcStyle}[2]{
  \expandafter\repcommand\expandafter{\csname TUserProcStyle@#1\endcsname}{#2}
}
\newcommand{\TProc}[1]{
  \ifcsname TUserProcStyle@\currentSort_#1\endcsname
    \myProc{#1}{\csname TUserProcStyle@\currentSort_#1\endcsname}{}
  \else
    \myProc{#1}{process}{}
  \fi
}

\newcommand{\repcommand}[2]{
  \providecommand{#1}{#2}
  \renewcommand{#1}{#2}
}
\newcommand{\THit}[5]{
  \path[hit] (#1) edge[#2] (#3#4);
  \expandafter\repcommand\expandafter{\csname TBounce@#3@#5\endcsname}{#4}
}
\newcommand{\TBounce}[4]{
  (#1\csname TBounce@#1@#3\endcsname) edge[#2] (#3#4)
}

\tikzstyle{plot}=[every path/.style={-}]
\tikzstyle{axe}=[gray,->,>=stealth']
\tikzstyle{ticks}=[-,font=\scriptsize,every node/.style={gray}]
\tikzstyle{mean}=[thick,-]
\tikzstyle{interval}=[line width=5pt,red,draw opacity=0.7,-]
\tikzstyle{bounce}=[->,densely dotted,>=stealth']
\tikzstyle{selfhit}=[min distance=5mm,curve to]
\tikzstyle{hitless graph}=[every edge/.style={draw,-}]
\tikzstyle{cross}=[preaction={draw=white,-,line width=6pt,shorten >=15pt}]
\tikzstyle{cross2}=[preaction={draw=white,-,line width=3pt,shorten >=10pt}]
\tikzstyle{grn}=[every node/.style={circle,draw,outer sep=2pt,minimum
size=20pt}]
\tikzstyle{elabel}=[draw=none,sloped,above=-5pt,outer sep=0,font=\scriptsize]
\tikzstyle{inh}=[->,>=|]
\tikzstyle{act}=[->,>=latex]

\tikzstyle{piproc}=[draw,circle,minimum size=22pt]

\tikzstyle{aS}=[every edge/.style={draw,->}]
\tikzstyle{Asol}=[draw,circle,minimum size=5pt,inner sep=0]
\tikzstyle{Aproc}=[draw,minimum width=16pt,minimum height=16pt,inner sep=1pt]
\tikzstyle{Aobj}=[]
\tikzstyle{Ainh}=[-|,shorten >= 1mm]

\title{Under-approximating Cut Sets for Reachability in Large Scale Automata Networks}

\author{Lo\"ic Paulev\'e\inst{1}, Geoffroy Andrieux\inst{2}, Heinz Koeppl\inst{1}}

\institute{
ETH Z\"urich, Switzerland.
\and
IRISA Rennes, France.
}

\begin{document}

\maketitle

\begin{abstract}
In the scope of discrete finite-state models of interacting components, 
we present a novel algorithm for identifying sets of local states of components whose activity is
necessary for the reachability of a given local state.
If all the local states from such a set are disabled in the model, the concerned reachability is impossible.

Those sets are referred to as cut sets and are computed from a particular abstract causality structure, 
so-called Graph of Local Causality, inspired from previous work and generalised here to finite
automata networks.
The extracted sets of local states form an under-approximation of the complete minimal
cut sets of the dynamics: there may exist smaller or additional cut sets for the given
reachability.

Applied to qualitative models of biological systems, such cut sets provide potential therapeutic targets that
are proven to prevent molecules of interest to become active, up to the correctness of the model.
Our new method makes tractable the formal analysis of very large scale networks, as illustrated by
the computation of cut sets within a Boolean model of biological pathways interactions
gathering more than 9000 components.
\end{abstract}

\section{Introduction}
\label{sec:intro}

With the aim of understanding and, ultimately, controlling physical systems, one generally constructs
dynamical models of the known interactions between the components of the system.
Because parts of those physical processes are ignored or still unknown, dynamics of such models have
to be interpreted as an over-approximation of the real system dynamics:
any (observed) behaviour of the real system has to have a matching behaviour in the abstract model, the
converse being potentially false.
In such a setting, a valuable contribution of formal methods on abstract models of physical systems
resides in the ability to prove the impossibility of particular behaviours.

Given a discrete finite-state model of interacting components, such as an automata network, 
we address here the computation of sets of local states of components that are necessary for 
reaching a local state of interest from a partially determined initial global state.
Those sets are referred to as \emph{cut sets}.
Informally, each trace leading to the reachability of interest has to involve, at one point, at
least one local state of a cut set.
Hence, disabling in the model all the local states referenced in one cut set should prevent the
occurrence of the concerned reachability from delimited initial states.

Applied to a model of a biological system where the reachability of interest is known to occur,
such cut sets provide potential coupled therapeutic targets to control the activity of a particular
molecule (for instance using gene knock-in/out).
The contrary implies that the abstract model is not an over-approximation of the concrete system.

\paragraph{Contribution.}
In this paper, we present a new algorithm to extract sets of local states that are necessary to
achieve the concerned reachability within a finite automata network.
Those sets are referred to as cut sets, and we limit ourselves to $N$-sets, \ie having a maximum
cardinality of $N$.

The finite automata networks we are considering are closely related to $1$-safe Petri nets
\cite{BC92} having mutually exclusive places.
They subsume Boolean and discrete networks \cite{Kauffman69,Thomas73,Richard10-AAM,BernotSemBRN},
synchronous or asynchronous, that are widely used for the qualitative modelling of biological
interaction networks.

A naive, but complete, algorithm could enumerate all potential candidate $N$-sets, disable each of
them in the model, and then perform model-checking to verify if the targeted reachability is still
verified.
If not, the candidate $N$-set is a cut set.
This would roughly leads to $m^N$ tests, where $m$ is the total number of local states in the
automata network. 
Considering that the model-checking within automata networks is PSPACE-complete \cite{Harel02}, this
makes such an approach intractable on large networks.

The proposed algorithm aims at being tractable on systems composed of a very large number of
interacting components, but each of them having a small number of local states.
Our method principally overcomes two challenges:
prevent a complete enumeration of candidate $N$-sets;
and prevent the use of model-checking to verify if disabling a set of local states break the
concerned reachability.
It inherently handles partially-determined initial states: the resulting cut $N$-set of local states are proven
to be necessary for the reachability of the local state of interest from \emph{any} of the supplied
global initial states.

The computation of the cut $N$-sets takes advantage of an abstraction of the formal model which
highlights some steps that are necessary to occur prior to the verification of a given
reachability properties.
This results in a causality structure called a \emph{Graph of Local Causality} (GLC), which is
inspired by \cite{PMR12-MSCS}, and that we generalise here to automata networks.
Such a GLC has a size polynomial with the total number of local states in the automata network, and
exponential with the number of local states within one automata.
Given a GLC, our algorithm propagates and combines the cut $N$-sets of the local states referenced
in this graph by computing unions or products, depending on the disjunctive or
conjunctive relations between the necessary conditions for their reachability.
The algorithm is proven to converge in the presence of dependence cycles.

In order to demonstrate the scalability of our approach, we perform the search for cut $N$-sets 
of processes within a very large Boolean model of biological processes relating more than 9000
components.
Despite the highly combinatorial dynamics, a prototype implementation manages to compute up to the
cut $5$-sets within a few minutes.
To our knowledge, this is the first time such a formal dynamical analysis has been performed on such
a large dynamical model of biological system.

\paragraph{Related work and limitations.}
Cut sets are commonly defined upon graphs as set of edges or vertices which, if removed,
disconnect a given pair of nodes \cite{SW86}.
For our purpose, this approach could be directly applied to the global transition graph (union of
all traces) to identify local states or transitions for which the remove would disconnect initial
states from the targeted states.
However, the combinatorial explosion of the state space would make it intractable for large
interacting systems.

The aim of the presented method is somehow similar to the generation of minimal cut sets in fault
trees \cite{FaultTreeReview,Tang04} used for reliability analysis, as the structure representing
reachability causality contains both \emph{and} and \emph{or} connectors.
However, the major difference is that we are here dealing with cyclic directed graphs which prevents the
above mentioned methods to be straightforwardly applied.

Klamt \textit{et al.} have developed a complete method for identifying minimal cut sets (also called intervention
sets) dedicated to biochemical reactions networks, hence involving cycling dependencies
\cite{Klamt04-Bioinformatics}.
This method has been later generalised to Boolean models of signalling networks
\cite{Samaga10-JCB}.
Those algorithms are mainly based on the enumeration of possible candidates, with techniques to
reduce the search space, for instance by exploiting symmetry of dynamics.
Whereas intervention sets of \cite{Klamt04-Bioinformatics,Samaga10-JCB} can contain either local
states or reactions, our cut sets are only composed of local states.

Our method follows a different approach than \cite{Klamt04-Bioinformatics,Samaga10-JCB} 
by not relying on candidate enumeration but computing the cut sets directly on an
abstract structure derived statically from the model, which should make tractable the analysis of
very large networks.
The comparison with \cite{Samaga10-JCB} is detailed in \pref{ssec:is}.

In addition, our method is generic to any automata network, but relies on an abstract interpretation
of dynamics which leads to under-approximating the cut sets for reachability:
by ignoring certain dynamical constraints, the analysis can miss several cut sets and output cut
sets that are not minimal for the concrete model.
Finally, although we focus on finding the cut sets for the reachability of only \emph{one} local
state, our algorithm computes the cut sets for the (independent) reachability of all local states
referenced in the GLC.

\paragraph{Outline.}
\pref{sec:structure} introduces a generic characterisation of the \emph{Graph of Local Causality} with
respect to automata networks;
\pref{sec:algorithm} states and sketches the proof of the algorithm for extracting a subset of
$N$-sets of local states necessary for the reachability of a given local state.
\pref{sec:application} discusses the application to systems biology by comparing with the
related work and applying our new method to a very large scale model of biological interactions.
Finally, \pref{sec:discuss} discusses the results presented and some of their possible extensions.

\paragraph{Notations.}
$\wedge$ and $\vee$ are the usual logical \textit{and} and \textit{or} connectors.
$[1;n] = \{1, \cdots, n\}$.
Given a finite set $A$, $\card A$ is the cardinality of $A$;
$\powerset(A)$ is the power set of $A$;
$\powerset^{\leq N}(A)$ is the set of all subsets of $A$ with cardinality at most $N$.
Given sets $A^1,\dots,A^n$,
$\bigcup_{i\in[1;n]} A^i$ is the union of those sets,
with the empty union $\bigcup_{\emptyset}\DEF\emptyset$;
and $A^1 \times \cdots \times A^n$ is the usual Cartesian product.
Given sets of sets $B^1,\dots,B^n \in \powerset(\powerset(A))$,
$\sprod_{i\in[1;n]} B^i \DEF B^1 \stimes \cdots \stimes B^n \in\powerset(\powerset(A))$ is the \emph{sets of sets product}
where
$\{ e_1, \dots, e_n\}\stimes \{ e'_1, \dots, e'_m \}
	\DEF \{ e_i \cup e'_j \mid i\in [1;n] \wedge j\in [1;m]\}$.
In particular $\forall (i,j)\in [1;n]\times[1;m]$, 
$B^i \stimes B^j = B^j \stimes B^i$ and $\emptyset \stimes B^i = \emptyset$.
The empty sets of sets product $\sprod_{\emptyset}\DEF\{ \emptyset \}$.
If $M:A \mapsto B$ is a mapping from elements in $A$ to elements in $B$,
$M(a)$ is the value in $B$ mapped to $a\in A$;
$M\{a \mapsto b\}$ is the mapping $M$ where $a\in A$ now maps to $b\in B$.

\section{Graph of Local Causality}
\label{sec:structure}

We first give basic definitions of automata networks, local state disabling, context and local state
reachability;
then we define the local causality of an objective (local reachability),
and the \emph{Graph of Local Causality}.
A simple example is given at the end of the section.

\subsection{Finite Automata Networks}

We consider a network of automata $(\Sigma,S,\L,T)$ which relates a finite number of interacting
finite state automata $\Sigma$ (\pref{def:cfsm}).
The global state of the system is the gathering of the local state of composing automata.
A transition can occur if and only if all the local states sharing a common transition label
$\ell\in L$ are present in the global state $s\in S$ of the system.
Such networks characterize a class of $1$-safe Petri Nets \cite{BC92} having groups of mutually exclusive places,
acting as the automata.
They allow the modelling of Boolean networks and their discrete generalisation, having either
synchronous or asynchronous transitions.

\begin{definition}[Automata Network $(\Sigma,S,\L,T)$]
\label{def:cfsm}
An automata network is defined by a tuple $(\Sigma,S,\L,T)$ where
\begin{itemize}
\item $\Sigma = \{ a, b, \dots, z\}$ is the finite set of automata identifiers;
\item For any $a\in\Sigma$, $S(a) = [1;k_a]$ is the finite set of local states of automaton $a$;
$S = \prod_{a\in\Sigma} [1;k_a]$ is the finite set of global states.
\item $\L=\{\ell_1, \dots, \ell_m\}$ is the finite set of transition labels;
\item $T = \{ a \mapsto T_a \mid a\in \Sigma \}$, where $\forall a\in\Sigma,
	T_a \subset [1;k_a] \times \L \times [1;k_a]$,
	is the mapping from automata to their finite set of local transitions.
\end{itemize}

We note $i\xrightarrow \ell j\in T(a) \EQDEF (i,\ell,j)\in T_a$
and
$a_i\xrightarrow \ell a_j\in T \EQDEF i\xrightarrow\ell j\in T(a)$.

$\forall\ell\in\L$, we note $\precond{\ell}\DEF\{a_i\mid a_i\xrightarrow\ell a_j\in T(a)\}$
and $\postcond{\ell}\DEF\{a_j\mid a_i\xrightarrow\ell a_j\in T(a)\}$.

The set of local states is defined as
$\Proc \DEF \{ a_i \mid a\in \Sigma \wedge i \in[1;k_a] \}$.

The global transition relation $\rightarrow \subset S\times S$ is defined as:
\begin{align*}
s\rightarrow s' \EQDEF 
\exists \ell\in\L: &
	\forall a_i\in\precond{\ell}, s(a) = a_i
\wedge \forall a_j\in\postcond{\ell}, s'(a) = a_j \\
\wedge & \forall b\in\Sigma, S(b)\cap\precond{\ell}=\emptyset
		\Rightarrow	s(b)=s'(b).
\end{align*}
\end{definition}

Given an automata network $\Sys=(\Sigma,S,\L,T)$ and a subset of its local states $ls\subseteq \Proc$,
$\Sys \disabling ls$ refers to the system where all the local states $ls$ have been disabled, i.e. they
can not be involved in any transition (\pref{def:disabling}).

\begin{definition}[Process Disabling]
\label{def:disabling}
Given $\Sys=(\Sigma,S,\L,T)$ and $ls\in\powerset(\Proc)$, 
$\Sys\disabling ls\DEF (\Sigma,S,\L',T')$ where
	$\L' = \{ \ell\in\L\mid ls\cap \precond\ell = \emptyset \}$
	and
	$T' = \{ a_i\xrightarrow{\ell} a_j\in T \mid \ell\in \L'\}$.
\end{definition}

From a set of acceptable initial states delimited by a \emph{context} $\ctx$ (\pref{def:ctx}),
we say a given local state $a_j\in\Proc$ is reachable if and only if
there exists a finite number of transitions in $\Sys$ leading to a global state where $a_j$ is
present (\pref{def:reachability}).

\begin{definition}[Context $\ctx$]
\label{def:ctx}
Given a network $(\Sigma,S,\L,T)$,
a context $\ctx$ is a mapping from each automaton $a\in\Sigma$ to a non-empty subset of its local
states:
$\forall a\in\Sigma, \ctx(a) \in \powerset(S(a)) \wedge \ctx(a)\neq\emptyset$.
\end{definition}

\begin{definition}[Process reachability]
\label{def:reachability}
Given a network $(\Sigma,S,\L,T)$ and a context $\ctx$,
the local state $a_j\in \Proc$ is \emph{reachable} from $\ctx$
if and only if 
$\exists s_0, \dots, s_m\in S$ such that
$\forall a\in \Sigma, s_0(a) \in \ctx(a)$,
and
$s_0 \rightarrow \cdots \rightarrow s_m$,
and
$s_m(a) = j$.
\end{definition}

\subsection{Local Causality}

Locally reasoning within one automaton $a$, the global reachability of $a_j$ from $\ctx$
can be expressed as the reachability of $a_j$ from a local state $a_i\in \ctx(a)$.
This local reachability specification is referred to as an \emph{objective}
noted $\obj{a_i}{a_j}$
(\pref{def:objective})

\begin{definition}[Objective]
\label{def:objective}
Given a network $(\Sigma,S,\L,T$),
the reachability of local state $a_j$ from $a_i$ is called an \emph{objective} and is denoted
$\obj{a_i}{a_j}$.
The set of all objectives is referred to as 
$\Obj\DEF\{ \obj{a_i}{a_j} \mid (a_i,a_j)\in\Proc \times \Proc \}$.
\end{definition}

Given an objective $P=\obj{a_i}{a_j}\in\Obj$, we define $\sol(P)$ the \emph{local causality} of $P$
(\pref{def:sol}):
each element $ls\in \sol(P)$ is a subset of local states, referred to as a (local) solution for $P$,
which are all involved at some times prior to the reachability of $a_j$.
$\sol(P)$ is sound as soon as the disabling of at least one local state in \emph{each} solution
makes the reachability of $a_j$ impossible from any global state containing $a_i$
(\pref{pro:sol}).
Note that if $\sol(P) = \{ \{a_i\} \cup ls^1, \dots, ls^m \}$ is sound,
$\sol'(P) = \{ ls^1, \dots, ls^m \}$ is also sound.
$\sol(\obj{a_i}{a_j})=\emptyset$ implies that $a_j$ can never be reached from $a_i$,
and $\forall a_i \in \Proc, \sol(\obj{a_i}{a_i}) \DEF \{ \emptyset \}$.

\begin{definition}
\label{def:sol}
$\sol : \Obj \mapsto \powerset(\powerset(\Proc))$
is a mapping from objectives to sets of sets of local states such that
$\forall P\in \Obj, \forall ls \in \sol(P), \nexists ls' \in \sol(P), ls\neq ls'$ such that
$ls' \subset ls$.
The set of these mappings is noted
$\Sol \DEF \{ \langle P,ls\rangle \mid ls \in \sol(P)\}$.
\end{definition}

\begin{property}[$\sol$ soundness]
\label{pro:sol}
$\sol(\obj{a_i}{a_j}) = \{ ls^1, \dots, ls^n \}$
is a sound set of solutions for the network $\Sys=(\Sigma,S,\L,T)$ if and only if
$\forall kls \in \sprod_{i\in[1;n]} ls^i$,
$a_j$ is not reachable in $\Sys\disabling kls$ from any state $s\in S$ such that $s(a)=i$.
\end{property}

In the rest of this paper we assume that \pref{pro:sol} is verified, and consider $\sol$ computation out of the scope of this paper.

Nevertheless, we briefly describe a construction of a sound $\sol(\obj{a_i}{a_j})$ for an automata
network $(\Sigma,S,\L,T)$; an example is given at the end of this section.
This construction generalises the computation of GLC from the Process Hitting framework, a
restriction of network of automata depicted in \cite{PMR12-MSCS}.
For each acyclic sequence $a_i \xrightarrow {\ell_1} \dots \xrightarrow {\ell_m} a_j$ of local
transitions in $T(a)$, 
and by defining 
$\f{ext}_a(\ell) \DEF \{ b_j \in\Proc\mid b_j\xrightarrow \ell b_k\in T, b\neq a\}$,
we set
$ls \in \sprod_{\ell\in \{ \ell_1, \dots, \ell_m \mid \f{ext}_a(\ell)\neq\emptyset \}} \f{ext}_a(\ell)
\Rightarrow ls\in\sol(\obj{a_i}{a_j})$, up to sursets removing.
One can easily show that \pref{pro:sol} is verified with such a construction.
The complexity of this construction is exponential in the number of local states of automata and
polynomial in the number of automata.
Alternative constructions may also provide sound (and not necessarily equal) $\sol$.

\subsection{Graph of Local Causality}

Given a local state $a_j\in \Proc$ and an initial context $\ctx$, the reachability of $a_i$ is
equivalent to the realization of any objective $\obj{a_i}{a_j}$, with $a_i\in\ctx(a)$.
By definition, if $a_j$ is reachable from $\ctx$, there exists $ls \in \sol(\obj{a_i}{a_j})$
such that, $\forall b_k\in ls$, $b_k$ is reachable from $\ctx$.

The (directed) \emph{Graph of Local Causality} (GLC, \pref{def:cwA}) relates this recursive reasoning
from a given set of local states $\w\subseteq \Proc$ by linking 
every local state $a_j$ to all objectives $\obj{a_i}{a_j}, a_i\in\ctx(a)$;
every objective $P$ to its solutions $\langle P, ls\rangle\in\Sol$;
every solution $\langle P, ls\rangle$ to its local states $b_k\in ls$.
A GLC is said to be valid if $\sol$ is sound for all referenced objectives (\pref{pro:cwA}).

\begin{definition}[Graph of Local Causality]
\label{def:cwA}
Given a context $\ctx$ and a set of local states $\w\subseteq \Proc$,
the \emph{Graph of Local Causality} (GLC)
$\cwA \DEF (\cwNodes, \cwEdges)$,
with
$\cwNodes \subseteq \Proc \cup \Obj \cup \Sol$
and
$\cwEdges \subseteq \cwNodes \times \cwNodes$,
is the smallest structure satisfying:
\begin{align*}
\w & \subseteq \cwNodes
\\
a_i\in\cwNodes\cap\Proc &\Leftrightarrow \{ (a_i,\obj{a_j}{a_i}) \mid a_j\in\ctx  \}\subseteq\cwEdges
\\
\obj{a_i}{a_j}\in\cwNodes\cap\Obj &\Leftrightarrow 
	\{ (\obj{a_i}{a_j},\langle\obj{a_i}{a_j},ls\rangle) \mid \langle\obj{a_i}{a_j},ls\rangle\in\Sol\}\subseteq\cwEdges
\\
\langle P, ls \rangle\in\cwNodes\cap\Sol &\Leftrightarrow
	\{ (\langle P,ls\rangle, a_i) \mid a_i\in ls \}\subseteq\cwEdges
	\enspace.
\end{align*}
\end{definition}

\begin{property}[Valid Graph of Local Causality]
A GLC $\cwA$ is \emph{valid} if, $\forall P\in \cwNodes\cap\Obj$, $\sol(P)$ is sound.
\label{pro:cwA}
\end{property}

This structure can be constructed starting from local states in $\w$ and by iteratively adding the imposed
children.
It is worth noticing that this graph can contain cycles.
In the worst case, $\card{\cwNodes} = \card{\Proc} + \card{\Obj} + \card{\Sol}$ and
$\card{\cwEdges} = \card{\Obj} + \card{\Sol} + \sum_{\langle P,ls\rangle\in \Sol} \card{ls}$.

\bigskip

\begin{example}
\label{ex:CFSM}
\pref{fig:graph-example} shows an example of GLC.
Local states are represented by boxed nodes and elements of $\Sol$ by small circles.

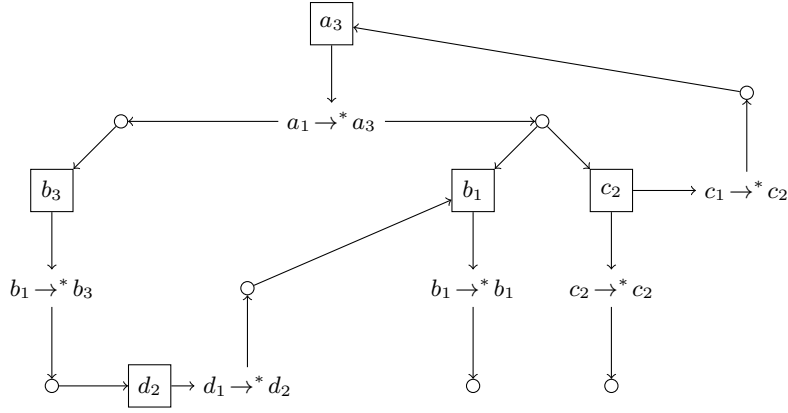
\begin{figure}[t]
\centering
\begin{tikzpicture}[aS,node distance=1.3cm,on grid=true]

\node[Aproc] (a2) {$a_3$};
\node[Aobj, below of=a2] (a02) {$\obj{a_1}{a_3}$};
\node[Asol, right of=a02, xshift=1.5cm] (sa02) {};
\node[Aproc, below right of=sa02] (c1) {$c_2$};
\node[Aproc, below left of=sa02] (b0) {$b_1$};

\node[Aobj, below of=c1] (c11) {$\obj{c_2}{c_2}$};
\node[Asol, below of=c11] (sc11) {};
\node[Aobj, right of=c1, xshift=0.5cm] (c01) {$\obj{c_1}{c_2}$};
\node[Asol, above of=c01] (sc01) {};

\node[Aobj, below of=b0] (b00) {$\obj{b_1}{b_1}$};
\node[Asol, below of=b00] (sb00) {};

\node[Asol, left of=a02, xshift=-1.5cm] (sa02b) {};
\node[Aproc, below left of=sa02b] (b2) {$b_3$};
\node[Aobj, below of=b2] (b02) {$\obj{b_1}{b_3}$};
\node[Asol, below of=b02] (sb02) {};
\node[Aproc, right of=sb02] (d1) {$d_2$};
\node[Aobj, right of=d1] (d01) {$\obj{d_1}{d_2}$};
\node[Asol, above of=d01] (sd01) {};

\path
	(a2) edge (a02)
	(a02) edge (sa02) edge (sa02b)
	(sa02) edge (b0) edge (c1)
	(b0) edge (b00) (b00) edge (sb00)
	(c1) edge (c11) edge (c01)
	(c11) edge (sc11)
	(c01) edge (sc01)
	(sc01) edge (a2)
	(sa02b) edge (b2)
	(b2) edge (b02) (b02) edge (sb02)
	(sb02) edge (d1)
	(d1) edge (d01)
	(d01) edge (sd01)
	(sd01) edge (b0)
;

\end{tikzpicture}

\caption{Example of Graph of Local Causality that is valid for the automata network defined in \pref{ex:CFSM}}
\label{fig:graph-example}
\end{figure}
For instance, such a GLC is valid for the following automata network $(\Sigma,S,\L,T)$,
with initial context $\ctx = \{ a\mapsto \{1\}; b\mapsto\{1\}; c\mapsto\{1,2\}; d\mapsto\{2\}\}$:
\begin{align*}
\Sigma & = \{ a,b,c,d \} & \L & = \{ \ell_1, \ell_2, \ell_3, \ell_4, \ell_5, \ell_6 \}\\
S(a) &= [1;3] & T(a)&=\{ 1 \xrightarrow{\ell_2} 2; 2 \xrightarrow{\ell_3} 3; 1 \xrightarrow{\ell_1} 3; 3 \xrightarrow{\ell_4} 2 \} 
\\
S(b) &= [1;3] & T(b)&=\{ 1 \xrightarrow{\ell_2} 2; 1 \xrightarrow{\ell_5} 3; 1 \xrightarrow{\ell_6} 1; 3 \xrightarrow{\ell_1} 2 \}
\\
S(c) &= [1;2] & T(c)&=\{ 1 \xrightarrow{\ell_4} 2; 2 \xrightarrow{\ell_3} 1 \}
\\
S(d) &= [1;2] & T(d)&=\{ 1\xrightarrow{\ell_6}2; 2\xrightarrow{\ell_5}1\}
\end{align*}

For example, within automata $a$, there are two acyclic sequences from $1$ to $3$: 
	$1 \xrightarrow{\ell_2} 2 \xrightarrow{\ell_3} 3$ and $1 \xrightarrow{\ell_1} 3$.
Hence, if $a_3$ is reached from $a_1$, then necessarily, one of these two sequences has to be used
(but not necessarily consecutively).
For each of these transitions, the transition label is shared by exactly one local state in another
automaton: $b_1$, $c_2$, $b_3$ for $\ell_2$, $\ell_3$, $\ell_1$, respectively.
Therefore, if $a_3$ is reached from $a_1$, then necessarily either both $b_1$ and $c_2$, or $b_3$ have been
reached before.
Hence $\sol(\obj{a_1}{a_3}) = \{ \{ b_1, c_2 \}, \{ b_3\}\}$ is sound, as disabling
either $b_1$ and $b_3$, or $c_2$ and $b_3$, would remove any possibility to reach $a_3$ from $a_1$.

\end{example}

\section{Necessary Processes for Reachability}
\label{sec:algorithm}

We assume a global sound GLC $\cwA=(\cwNodes,\cwEdges)$, with the usual accessors
for the direct relations of nodes:
\begin{align*}
\childs&: \cwNodes \mapsto \powerset(\cwNodes)
&
\parents&: \cwNodes \mapsto \powerset(\cwNodes)
\\
\childs(n) &\DEF \{ m\in\cwNodes\mid (n,m)\in\cwEdges\} 
&
\parents(n) &\DEF \{ m\in\cwNodes\mid (m,n)\in\cwEdges\}
\end{align*}

Given a set of local states $\Obs\subseteq\Proc$, this section introduces an algorithm computing upon
$\cwA$ the set $\val(a_i)$ of minimal cut $N$-sets of local states in $\Obs$ that are necessary for
the independent reachability of each local state $a_i\in\Proc\cap\cwNodes$.
The minimality criterion actually states that $\forall ls\in\val(a_i)$, there is no different
$ls'\in\val(a_i)$ such that $ls'\subset ls$.

Assuming a first valuation $\val$ (\pref{def:val}) associating to each node a set of (minimal) cut
$N$-sets, the set of cut $N$-sets for the node $n$ can be refined following $\update(\val,n)$
(\pref{def:update}):
\begin{itemize}
\item 
if $n$ is a solution $\langle P, ls\rangle\in\Sol$,
it is sufficient to prevent the reachability of \emph{any} local state in $ls$ to cut $n$;
therefore, the cut $N$-sets results from the union of the cut $N$-sets of $n$ children (all local
states).
\item 
If $n$ is an objective $P\in\Obj$, all its solutions (in $\sol(P)$) have to be cut in order to
ensure that $P$ is not realizable:
hence, the cut $N$-sets result from the product of children cut $N$-sets (all solutions).
\item 
If $n$ is a local state $a_i$, it is sufficient to cut all its children (all objectives) to prevent the reachability of $a_i$
from any state in the context $\ctx$.
In addition, if $a_i\in\Obs$, $\{a_i\}$ is added to the set of its cut $N$-sets.
\end{itemize}

\begin{definition}[Valuation $\val$]
\label{def:val}
A \emph{valuation} $\val: \cwNodes \mapsto \SetNSets{\Obs}$
is a mapping from each node of $\cwA$ to a set of $N$-sets of local states.
$\Val$ is the set of all valuations.
$\val_0\in \Val$ refers to the valuation such that $\forall n\in\cwNodes, \val_0(n) = \emptyset$.
\end{definition}

\begin{definition}[$\update : \Val \times \cwNodes \mapsto \Val$]
\label{def:update}
\begin{align*}
\update(\val, n) & \DEF
\begin{cases}
\val\{n\mapsto\simplN{\bigcup_{m\in\childs(n)} \val(m)} \} & \text{if }n\in \Sol \\
\val\{n\mapsto\simplN{\sprod_{m\in\childs(n)} \val(m)} \} & \text{if }n\in \Obj \\
\val\{n\mapsto\simplN{\sprod_{m\in\childs(n)} \val(m)} \}
		& \text{if }n\in \Proc \setminus \Obs\\
\val\{n\mapsto\simplN{\{\{a_i\}\} \cup \sprod_{m\in\childs(n)} \val(m)} \}
		& \text{if }n\in \Proc \cap \Obs
\end{cases}
\end{align*}
where 
$\simplN{\{e_1, \dots, e_n\}} \DEF \{ e_i \mid i\in[1;n] \wedge \card{e_i}\leq N 
									\wedge \nexists j\in[1;n], j\neq i, e_j\subset e_i \}$,
$e_i$ being sets, $\forall i\in[1;n]$.
\end{definition}

Starting with $\val_0$, one can repeatedly apply $\update$ on each node of $\cwA$ to refine its
valuation.
Only nodes where one of their children value has been modified should be considered for updating.

Hence, the order of nodes updates should follow the topological order of the GLC, where children
have a lower rank than their parents (i.e., children are treated before their parents).
If the graph is actually acyclic, then it is sufficient to update the value of each node only once.
In the general case, \ie in the presence of Strongly Connected Components (SCCs) --- nodes belonging
to the same SCC have the same rank ---, the nodes within a SCC have to be iteratively updated until
the convergence of their valuation.

\pref{alg:flood} formalizes this procedure where
$\rank(n)$ refers to the topological rank of $n$, as it can be derived from 
Tarjan's strongly connected components algorithm \cite{Tarjan72}, for example.
The node $n\in\cwNodes$ to be updated is selected as being the one having the least rank amongst the
nodes to update (delimited by $\M$).
In the case where several nodes with the same lowest rank are in $\M$, they can be either
arbitrarily or randomly picked.
Once picked, the value of $n$ is updated.
If the new valuation of $n$ is different from the previous, the parents of $n$ are added to the
list of nodes to update (lines 6-8 in \pref{alg:flood}).

\begin{algorithm}[t]
\algrenewcommand{\algorithmiccomment}[1]{\hskip1em /\textit{#1}/}
\begin{algorithmic}[1]
\State $\M \gets \cwNodes$
\State $\val \gets \val_0$
\While{$\M \neq \emptyset$}
\State $n \gets \arg\min_{m\in\M} \{ \rank(m) \}$
\State $\M \gets \M \setminus \{ n \}$
\State $\val' \gets \update(\val, n)$
\If{$\val'(n) \neq \val(n)$}
\State $\M \gets \M \cup \parents(n)$
\EndIf
\State $\val \gets \val'$
\EndWhile
\State \Return $\val$
\end{algorithmic}
\caption{\textproc{\floodname}}
\label{alg:flood}
\end{algorithm}

\pref{lem:halt} states the convergence of \pref{alg:flood} and \pref{thm:correctness} its
correctness:
for each local state $a_i\in \cwNodes\cap\Proc$, 
each set of local states $kls\in\val(a_i)$ (except $\{a_i\}$ singleton) references the local states that are
all necessary to reach prior to the reachability of $a_i$ from any state in $\ctx$.
Hence, if all the local states in $kls$ are disabled in $\Sys$, $a_i$ is not reachable from any state in $\ctx$.

\begin{lemma}
\label{lem:halt}
\textproc{\floodname}
always terminates.
\end{lemma}
\begin{proof}
Remarking that $\SetNSets{\Obs}$ is finite, defining a partial ordering such that
$\forall v,v'\in\SetNSets{\Obs}, v \succeq v' \EQDEF \simplN{v} = \simplN{v\cup v'}$,
and noting $\val^k \in \Val$ the valuation after $k$ iterations of the algorithm,
it is sufficient to prove that $\val^{k+1}(n) \succeq \val^k(n)$.
Let us define $v_1, v_2, v'_1, v'_2\in\SetNSets{\Obs}$ such that $v_1\succeq v'_1$ and 
$v_2\succeq v'_2$.
We can easily check that $v_1\cup v_2 \succeq v'_1\cup v'_2$ (hence proving the case when
$n\in\Sol$).
As $\simplN{v_1}=\simplN{v_1\cup v'_1} \Leftrightarrow \forall e'_1 \in v'_1, \exists e_1\in v_1:
e_1\subseteq e'_1$,
we obtain that $\forall (e'_1,e'_2)\in v'_1 \times v'_2, \exists (e_1,e_2)\in v_1\times v_2: 
e_1\subseteq e'_1 \wedge e_2\subseteq e'_2$.
Hence $e_1\cup e_2\subseteq e'_1\cup e'_2$, therefore
	$\simplN{v_1\stimes v_2 \cup v'_1\stimes v'_2} = \simplN{v_1\stimes v_2}$,
i.e. $v_1\stimes v_2\succeq v'_1 \stimes v'_2$;
which proves the cases when $n\in \Obj \cup \Proc$.
\end{proof}

\begin{theorem}
\label{thm:correctness}
Given a GLC $\cwA=(\cwNodes,\cwEdges)$ which is sound for the automata network $\Sys$,
the valuation $\val$ computed by $\textproc{\floodname}$ verifies:
$\forall a_i \in \Proc\cap\cwNodes,
	\forall kls\in \val(a_i) \setminus \{\{a_i\}\}$, 
		$a_j$ is not reachable from $\ctx$ within $\Sys \disabling kls$.
\end{theorem}
\begin{proof}
By recurrence on the valuations $\val$: the above property is true at each iteration of the algorithm.
\end{proof}

\begin{example}
\pref{tab:example} details the result of the execution of \pref{alg:flood} on the GLC defined in
\pref{fig:graph-example}.
Nodes receive a topological rank, identical ranks implying the belonging to the same SCC.
The (arbitrary) scheduling of the updates of nodes within a SCC follows the order in the table.
In this particular case, nodes are all visited once, as
$\val(\langle \obj{c_2}{c_2},\emptyset\rangle) \stimes \val(\langle \obj{c_1}{c_2}, \{a_3\}\rangle)
= \emptyset$ (hence $\update(\val, c_2)$ does not change the valuation of $c_2$).
Note that in general, several iterations of $\update$ may be required to reach a fixed point.

\begin{table}[t]
\centering
\begin{tabular}{|c|c|l|}
\hline
Node & $\f{rank}$ & $\val$ 
\\\hline
$\langle \obj{b_1}{b_1}, \emptyset\rangle$ & 1 & $ \emptyset $
\\
$\obj{b_1}{b_1}$ & 2 & $ \emptyset $
\\
$b_1$ & 3 & $\{\{ b_1 \}\}$
\\
$\langle \obj{d_1}{d_2}, \{b_1\}\rangle$ & 4 & $\{\{ b_1 \}\}$
\\
$\obj{d_1}{d_2}$ & 5 & $\{\{ b_1 \}\}$
\\
$d_2$ & 6 & $\{\{b_1\},\{d_2\}\}$
\\
$\langle \obj{b_1}{b_3}, \{d_2\}\rangle$ & 7 & $\{\{b_1\},\{d_2\}\}$
\\
$\obj{b_1}{b_3}$ & 8 & $\{\{b_1\},\{d_2\}\}$
\\
$b_3$ & 9 & $\{\{b_1\},\{b_3\},\{d_2\}\}$
\\
$\langle \obj{a_1}{a_3}, \{b_3\}\rangle$ & 10 & $\{\{b_1\},\{b_3\},\{d_2\}\}$
\\
$\langle \obj{c_2}{c_2}, \emptyset\rangle$ & 11 & $\emptyset$
\\
$\obj{c_2}{c_2}$ & 12 & $\emptyset$
\\
$c_2$ & 13 & $\{\{c_2\}\}$
\\
$\langle \obj{a_1}{a_3}, \{b_1,c_2\}\rangle$ & 13 & $\{\{b_1\},\{c_2\}\}$
\\
$\obj{a_1}{a_3}$ & 13 & $\{\{b_1\},\{b_3,c_2\},\{c_2,d_2\}\}$
\\
$a_3$ & 13 & $\{ \{a_3\},\{b_1\},\{b_3,c_2\},\{c_2,d_2\}\}$
\\
$\langle \obj{c_1}{c_2}, \{a_3\}\rangle$ & 13 & $\{ \{a_3\},\{b_1\},\{b_3,c_2\},\{c_2,d_2\}\}$
\\\hline
\end{tabular}
\caption{Result of the execution of \pref{alg:flood} on the GLC in \pref{fig:graph-example}}
\label{tab:example}
\end{table}

\end{example}

It is worth noticing that the GLC abstracts several dynamical constraints in the underlying automata
networks, such as the ordering of transitions, or the synchronous updates of the global state.
In that sense, GLC over-approximates the dynamics of the network, and the resulting cut sets are
under-approximating the complete cut sets of the concrete model.

\section{Application to Systems Biology}
\label{sec:application}

Automata networks, as presented in \pref{def:cfsm}, subsume Boolean and discrete networks,
synchronous and asynchronous, that are widely used for the qualitative modelling of dynamics of
biological networks \cite{Kauffman69,Thomas73,deJong02,Richard10-AAM,BernotSemBRN,egfr104,Laubenbacher11}.

A cut set, as extracted by our algorithm, informs that at least one of the component in the cut
set has to be present in the specified local state in order to achieve the wanted reachability.
A local state can represent, for instance, an active transcription factor or the absence of a certain protein.
It provides potential therapeutic targets if the studied reachability is involved in a disease by preventing
all the local states of a cut set to act, for instance using gene knock-out or knock-in techniques.

We first discuss and compare our methodology with the \emph{intervention sets} analysis within
biological models developed by S. Klamt \textit{et al.}, and provide some benchmarks on a few
examples.

Thanks to the use of the intermediate GLC and to the absence of candidate enumeration, our new method
makes tractable the cut sets analysis on very large models.
We present a recent application of our results to the analysis of a very large scale Boolean
model of biological pathway interactions involving 9000 components.
To our knowledge, this is the first attempt of a formal dynamical analysis on such a large scale model.

\subsection{Related Work}
\label{ssec:is}

The general related work having been discussed in \pref{sec:intro}, we deepen here the comparison of
our method with the closest related work: the analysis of \emph{Intervention Sets} (IS) \cite{Samaga10-JCB}.
In the scope of Boolean models of signalling networks, an IS is a set of local states such that
forcing the components to stick at these local states ensures that the system always reaches a fixed
point (steady state) where certain target components have the desired state.
Their method is complete, i.e., all possible ISs are computed; and contrary to our, allow the
specification of more than one local state for the target state of the intervention.

Nevertheless, the semantics and the computation of ISs have some key differences with our computed
cut sets.
First, they focus only on the reachability of (logical) steady states, which is a stronger condition
than the transient reachability that we are considering.
Then, the steady states are computed using a three-valued logic which allows to cope with undefined
(initial) local states, but which is different from the notion of context that we use in this paper for
specifying the initial condition.

Such differences make difficult a proper comparison of inferred cut sets.
We can however expect that any cut sets found by our method has a corresponding IS in the scope of
Boolean networks when the initial context actually specifies a single initial state.

To give a practical insight on the relation between the two methods, we compare the results for two
signalling networks, both between a model specified with CellNetAnalyser \cite{CNA} to compute ISs
and a model specified in the Process Hitting framework, a particular restriction of asynchronous
automata networks \cite{PMR10-TCSB}, to compute our cut sets.
Process Hitting models have been built in order to over-approximate the dynamics considered for the
computation of ISs%
\footnote{Models and scripts available at \url{http://loicpauleve.name/cutsets.tbz2}}.

\paragraph{Tcell.}
Applied to a model of the T-cell receptor signalling between 40 components \cite{Klamt06}, we are
interested in preventing the activation of the transcription factor $AP1$.
For an instance of initial conditions, and limiting the computations to $3$-sets,
$31$ ISs have been identified ($28$ $1$-sets, $3$ $2$-sets, $0$ $3$-set),
whereas our algorithm found $29$ cut sets ($21$ $1$-sets and $8$ $2$-sets),
which are all matching an IS ($23$ are identical, $6$ strictly including ISs).
ISs are computed in $0.69$s while our algorithm under-approximates the cut sets in $0.006$s.
Different initial states give comparable results.

\paragraph{Egfr.}
Applied to a model of the epidermal growth factor receptor signalling pathway of 104 components
\cite{egfr104}, we are interested in preventing the activation of the transcription factor $AP1$.
For an instance of initial conditions, and limiting the computations to $3$-sets,
$25$ ISs have been identified ($19$ $1$-sets, $3$ $2$-sets, $3$ $3$-sets),
whereas our algorithm found $14$ cut sets ($14$ $1$-sets), which are all included in the ISs.
ISs are computed in $98$s while our algorithm under-approximates the cut sets in $0.004$s.
Different initial states give comparable results.

\medskip

As expected with the different semantics of models and cut sets, resulting ISs matches all the cut
sets identified by our algorithm, and provides substantially more sets.
The execution time is much higher for ISs as they rely on candidate enumeration in order to provide
complete results, whereas our method was designed to prevent such an enumeration but
under-approximates the cut sets.

\medskip

In order to appreciate the under-approximation done by our method at a same level of abstraction and
with identical semantics, we compare the cut sets identified by our algorithm with the cut sets
obtained using a naive, but complete, computation.
The naive computation enumerates all cut set candidates and, for each of them, disable the local
states in the model and perform model-checking to verify if the target local state is still reachable.
In the particular case of these two models, and limiting the cut sets to $3$ and $2$-sets
respectively for the sake of tractability, no additional cut set has been uncovered by the complete
method.
Such a good under-approximation could be partially explained by the restrictions imposed on the
causality by the Process Hitting framework, making the GLC a tight over-approximation of the
dynamics \cite{PMR12-MSCS}.

\subsection{Very Large Scale Application to Pathway Interactions}

In order to support the scalability of our approach, we apply the proposed algorithm to a very large
model of biological interactions, actually extracted from the PID database \cite{PID} referencing 
various influences (complex formation, inductions (activations) and inhibitions, transcriptional
regulation, etc.) between more than 9000 biological components (proteins, genes, ions, etc.).

Amongst the numerous biological components, the activation of some of them are known to control key
mechanisms of the cell dynamics.
Those activations are the consequence of intertwining signalling pathways and depend on
the environment of the cell (represented by the presence of certain \textit{entry-point} molecules).
Uncovering the environmental and intermediate components playing a major role in these
signalling dynamics is of great biological interest.

The full PID database has been interpreted into the Process Hitting framework, a subclass of
asynchronous automata networks, from which the derivation of the GLC has been addressed in previous
work \cite{PMR12-MSCS}.
The obtained model gathers components representing either biological entities modelled as boolean
value (absent or present), or logical complexes.
When a biological component has several competing regulators, the precise cooperations are not
detailed in the database, so we use of two different interpretations:
all (resp. one of) the activators and none (resp. all but one of) the inhibitors have to be present
in order to make the target component present.
This leads to two different discrete models of PID that we refer to as 
\texttt{whole\_PID\_AND} and \texttt{whole\_PID\_OR}, respectively.

Focusing on \texttt{whole\_PID\_OR}, the Process Hitting model relates more than 21000 components,
either biological or logical, containing between 2 and 4 local states.
Such a system could actually generate $2^{33874}$ states.
3136 components act as environment specification, which in our boolean interpretation
leads to $2^{3136}$ possible initial states, assuming all other components start in the absent state.

We focus on the (independent) reachability of active SNAIL transcription factor, involved in the
epithelial to mesenchymal transition \cite{emt}, and of active p15INK4b and p21CIP1 cyclin-dependent
kinase inhibitors involved in cell cycle regulation \cite{tgf}.
The GLC relates $20045$ nodes, including $5671$ component local states (biological or logical);
it contains 6 SCCs with at least 2 nodes, the largest being composed of $10238$ nodes
and the others between $20$ and $150$.

\pref{tab:exec} shows the results of a prototype implementation%
\footnote{Implemented as part of the \textsc{Pint} software --
\url{http://process.hitting.free.fr}\\
Models and scripts available at \url{http://loicpauleve.name/cutsets.tbz2}
}
of \pref{alg:flood} for the search of up to the $6$-sets of biological component local states.
One can observe that the execution time grows very rapidly with $N$ compared to the number of
visited nodes.
This can be explained by intermediate nodes having a large set of cut $N$-sets leading to a costly
computation of products.

\begin{table}[t]
\centering
\begin{tabular}{|l|c|c|c|c|c|c|}
\hline
\multirow{2}{*}{Model} &
\multirow{2}{*}{N} &
\multirow{2}{*}{Visited nodes} &
\multirow{2}{*}{Exec. time} &
\multicolumn{3}{c|}{Nb. of resulting N-sets}
\\\cline{5-7}
& & & & 
SNAIL$_1$ & p15INK4b$_1$ & p21CIP1$_1$
\\\hline
\multirow{6}{*}{whole\_PID\_OR}
& 1 & 29022 & 0.9s & 1 & 1 & 1 \\\cline{2-7}
& 2 & 36602 & 1.6s & +6 & +6 & +0 \\\cline{2-7}
& 3 & 44174 & 5.4s & +0 & +92 & +0 \\\cline{2-7}
& 4 & 54322 & 39s & +30 & +60 & +0 \\\cline{2-7}
& 5 & 68214 & 8.3m & +90 & +80 & +0 \\\cline{2-7}
& 6 & 90902 & 2.6h & +930 & +208 & +0 
\\\hline
\end{tabular}
\caption{
Number of nodes visited and execution time of the search for cut $N$-sets of 3 local states.
For each N, only the number of additional N-sets is displayed.}
\label{tab:exec}
\end{table}

While the precise biological interpretation of identified $N$-sets is out of the scope of this paper,
we remark that the order of magnitude of the number of cut sets can be very different
(more than $1000$ cut $6$-sets for SNAIL; none cut $6$-sets for p21CIP1, except the gene that
produces this protein).
It supports a notion of robustness for the reachability of components, where the less cut sets, the
more robust the reachability to various perturbations.

Applied to the \texttt{whole\_PID\_AND} model, our algorithm find in general much more cut $N$-sets,
due to the conjunctive interpretation.
This brings a significant increase in the execution time: the search up to the cut
$5$-sets took 1h, and the $6$-sets leads to an out-of-memory failure.

\smallskip

Because of the very large number of components involved in this model, the naive exact algorithm
consisting in enumerating all possible $N$-sets candidates and verifying the concerned reachability
using model-checking is not tractable.
Similarly, making such a model fit into other frameworks, such as CellNetAnalyser (see previous
sub-section) is a challenging task, and might be considered as future work.

\section{Discussion}
\label{sec:discuss}

We presented a new method to efficiently compute cut sets for the reachability of a local state of a
component within networks of finite automata from any state delimited by a provided so-called
context.
Those cut sets are sets of automata local states such that disabling the activity of all local
states of a cut set guarantees to prevent the reachability of the concerned local state.
Automata networks are commonly used to represent the qualitative dynamics of interacting biological
systems, such as signalling networks.
The computation of cut sets can then lead to propose potential therapeutic targets that have been
formally identified from the model for preventing the activation of a particular molecule.

The proposed algorithm works by propagating and combining the cut sets of local states along a 
\emph{Graph of Local Causality} (GLC), that we introduce here upon automata networks.
A GLC relates the local states that are necessary to occur prior to the reachability
of the concerned local state.
Several constructions of a GLC are generally possible and depend on the semantics of the
model.
We gave an example of such a construction for automata networks.
That GLC has a size polynomial in the total number of local states in the network, and exponential
in the number of local states within one automaton.
Note that the core algorithm for computing the cut sets only requires as input a GLC satisfying
a soundness property that can be easily extended to discrete systems that are more general than the
automata networks considered here.

The computed cut sets form an under-approximation of the complete cut sets as the GLC abstracts
several dynamical constraints from the underlying concrete model.
Our algorithm does not rely on a costly enumeration of the potential sets of candidates, 
and thus aim at being tractable on very large networks.

A prototype implementation of our algorithm has been successfully applied to the extraction of
cut sets from a Boolean model of a biological system involving more than 9000 interacting components.
To our knowledge this is the first attempt of such a dynamical analysis for such large biological models.
We note that most of the computation time is due to products between large sets of cut $N$-sets.
To partially address this issue, we use of prefix trees to represent set of sets
on which we have specialized operations to stick to sets of $N$-sets (\pref{suppl:bdd}).
There is still room for improvement as our prototype does not implement any caching or variable
re-ordering.

\medskip

The work presented in this paper can be extended in several ways, notably with a 
\emph{posterior enlarging of the cut sets}.
Because the algorithm computes the cut $N$-sets for each node in the GLC,
it is possible to construct \textit{a posteriori} cut sets with a greater cardinality by chaining them.
For instance, let $kps\in\val(a_i)$ be a cut $N$-set for the reachability of $a_i$,
for each $b_j\in kps$ and $kps'\in\val(b_j)$, $(kps\setminus\{b_j\})\cup kps' $ is a cut set for $a_i$.
In our biological case study, this method could be recursively applied until cut sets are composed of
states of automata only acting for the environmental input.

With respect to the defined computation of cut $N$-sets, one could also derive \emph{static
reductions} of the GLC.
Indeed, some particular nodes and arcs of the GLC can be removed without affecting the final
valuation of nodes. 
A simple example are nodes representing objectives having no solution: such nodes can be safely
removed as they bring no candidate $N$-sets for parents processes.
These reductions conduct to both speed-up of the proposed algorithm but also to more compact
representations for the reachability causality. 

Finally, future work is considering the identification of necessary transitions (reactions), in addition to
local states, that are necessary for a reachability to occur.
The introduction of additional dynamical constraints in the GLC, such as conflicts or time scales,
would also help to increase the number of cut sets identifiable by such abstract interpretation
techniques.

\bigskip

\noindent
\textbf{Acknowledgements.}
LP and GA acknowledge the partial support of the French National Agency for Research
(ANR-10-BLANC-0218 BioTempo project).

\bibliographystyle{splncs}
\bibliography{bibliography,pauleve}

\appendix
\clearpage
\section{Implementation of Sets of Minimal N-Sets}
\label{suppl:bdd}

This appendix gives some details on the data structure we developed to efficiently manipulate sets
of minimal $N$-sets, i.e., sets of $N$-sets containing no sursets.
The data structure is similar to prefix trees, on which operations have been design to
perform union, product and simplification (minimisation) of sets of $N$-sets.
An OCaml\footnote{\url{http://caml.inria.fr}}
implementation of these routines is available at
\url{http://code.google.com/p/pint/source/browse/pintlib/kSets.ml}.

Given a (possibly infinite) set of totally ordered elements, such as integers,
the data structure is a forest where leafs are either $\bot$ or $\top$, and intermediate nodes are
elements.
Each path from any root to any leaf form a strictly increasing sequence of elements.
The maximum height of the forest is $N+1$.
\pref{fig:nsets} gives an example of such a structure.

\begin{figure}
\centering
\begin{tikzpicture}

\node at (2,0) (a2) {$2$};
\node[below of=a2] (b5) {$5$};
\node[below of=b5] (c5t) {$\top$};

\node at (0.5,0) (a1) {$1$};

\node[left of=b5] (b3) {$3$};
\node[below of=b3] (c3t) {$\top$};

\node[left of=b3] (b2) {$2$};
\node[below of=b2] (c4) {$4$};
\node[below of=c4] (d4t) {$\top$};

\path[->]
	(a1) edge (b2)
	(b2) edge (c4)
	(c4) edge (d4t)
	(a1) edge (b3)
	(b3) edge (c3t)
	(a2) edge (b5)
	(b5) edge (c5t)
;

\end{tikzpicture}
\caption{Representation of set of sets $\{ \{1,2,4\}, \{1,3\}, \{2,5\} \}$}
\label{fig:nsets}
\end{figure}
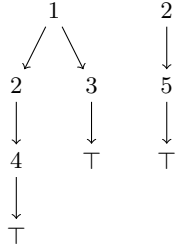

\subsubsection*{Data structure}
\def\G{\mathcal D}

An instance of the data structure is either $\top$, acting for the set containing the empty set
$(\{\emptyset\})$, or $\bot$, acting for the empty set, or an (ordered) associative map from
(prefix) elements to other instances of the data structure.
This is summarised by the following definition:
\begin{equation}
\G ::= \top \mid \bot \mid \{ i_1 \mapsto \G_1, \cdots, i_k \mapsto \G_k \}
\label{eq:nsets}
\end{equation}
with $k\geq 1$ and $i_1 < \cdots < i_k$.
Given $\G = \{ i_1\mapsto \G_1, \dots, i_k\mapsto \G_k\}$, $i_1,\dots,i_k$ are the \emph{prefixes},
and $\G_1,\dots,\G_k$ their corresponding \emph{suffixes}.
We also note $\forall j\in[1;k], \G(i_j) \DEF \G_j$.
As mentioned, $\bot$ acts as the empty set, so any prefix mapped to an empty set can be removed:
\begin{align*}
\{ i_1 \mapsto \bot \} &\equiv \bot
\\
\{ i_1 \mapsto \G_1,  L, i_j\mapsto \bot, R \} &\equiv \{ i_1\mapsto \G_1, L, R \}
\end{align*}
Hereafter, we assume that this removing is done implicitly.

\begin{example} The set of sets $\{ \{1,2,4\}, \{1,3\}, \{2,5\} \}$ is encoded as
\[\{ 1 \mapsto \{ 2 \mapsto \{ 4 \mapsto \top \}, 3 \mapsto \top \}, 2 \mapsto \{ 5 \mapsto \top
\}\}\]
which corresponds to the forest in \pref{fig:nsets}.
\end{example}

We first describe two helper functions on top of $\G$ that will be used for the operations.

\paragraph{\textproc{inds}($\G$)}
Given a data structure $\G$, the \textproc{inds} function returns the sequence of prefix elements
in the \emph{reverse order}. If $\G$ is either $\top$ or $\bot$, the empty sequence is returned.

\paragraph{\textproc{up}($\G,h$)}
Given a data structure $\G$ and a level $h$ (initially $1$), \textproc{up} removes all the sets in $\G$ that can not
be extended with one additional elements, i.e., all the sets with at least $N$ elements.

\begin{figure}
\hrule
\medskip

\begin{small}
\begin{algorithmic}
\Function{inds}{$\G$}
\If{$\G \in \{\top, \bot\}$}
 \Return $[]$
\ElsIf{$\G \equiv \{ i_1 \mapsto \G_1, \cdots, i_k \mapsto \G_k \}$}
\State \Return $[i_k, \dots, i_1]$
\EndIf
\EndFunction
\end{algorithmic}
\end{small}

\medskip
\hrule
\medskip

\begin{small}
\begin{algorithmic}
\Function{up}{$\G,h$}
\If{$\G \in \{\top, \bot\}$}
\Return $\G$
\ElsIf{$h \geq N$}
\Return $\bot$
\Else
\For{$i\gets \Call{inds}{\G}$}
\State $\G(i) \gets \Call{up}{\G(i),h+1}$
\EndFor
\State \Return $\G$
\EndIf
\EndFunction
\end{algorithmic}
\end{small}

\medskip
\hrule
\end{figure}

\subsubsection*{Union and product}
The \textproc{union} and \textproc{product} (realizing the $\stimes$ operator described in
\pref{sec:intro}) operations guarantee the ordering between prefixes, and that no set has
cardinality strictly greater than $N$.

\paragraph{\textproc{union}($\G_a,\G_b$)}
Given two set of sets, this function merges the prefixes of sets.

\paragraph{\textproc{product}($\G_a,\G_b,h$)}
Given two set of sets at level $h$, the product is computed as follows:
if two sets share the same prefix, the product results from the product of suffixes;
if two sets have different prefixes, the one having the highest prefix is augmented by one level
(if possible), and its product is computed with the suffix of the other. 
The result is the suffix of the lowest prefix.

\begin{figure}
\hrule\medskip

\begin{small}
\begin{algorithmic}
\Function{union}{$\G_a,\G_b$}
\If{$\G_a = \top$ \textbf{or} $\G_b = \top$}
	\Return $\top$
\ElsIf{$\G_a = \bot$}
	\Return $\G_b$
\ElsIf{$\G_b=\bot$}
	\Return $\G_a$
\Else
\State $\G \gets \G_b$
\For{$i\gets\Call{inds}{\G_a}$}
	\If{$i \in \Call{inds}{\G_b}$}
		\State $\G(i) \gets\Call{union}{\G_a(i),\G_b(i)}$
	\Else
		\State $\G(i) \gets \G_a(i)$
	\EndIf
\EndFor
\State \Return $\G$
\EndIf
\EndFunction
\end{algorithmic}
\end{small}

\medskip
\hrule
\medskip

\begin{small}
\begin{algorithmic}
\Function{product}{$\G_a,\G_b,h$}
\If{$\G_a = \top$}
	\Return $\G_b$
\ElsIf{$\G_b=\top$}
	\Return $\G_a$
\Else
\State $\G \gets \bot$
\For{$i\gets\Call{inds}{\G_a}$}
\For{$j\gets\Call{inds}{\G_b}$}
	\If{$i = j$}
		\State $\G_{ij}\gets\Call{product}{\G_a(i),\G_b(j),h+1}$
	\ElsIf{$i > j$}
		\State $\G_{a/i} \gets \{ i\mapsto \Call{up}{\G_a(i),h+1}$
		\State $\G_{ij}\gets\Call{product}{\G_{a/i}, \G_b(j), h+1}$
	\ElsIf{$i < j$}
		\State $\G_{b/j} \gets \{ j\mapsto \Call{up}{\G_b(j),h+1}$
		\State $\G_{ij}\gets\Call{product}{\G_a(i), \G_{b/j}, h+1}$
	\EndIf
	\State $\G\gets\Call{union}{\G, \G_{ij}}$
\EndFor
\EndFor
\State \Return $\G$
\EndIf
\EndFunction
\end{algorithmic}
\end{small}

\medskip
\hrule
\end{figure}

\subsubsection*{Sursets simplification (minimisation)}
The above operations do not guarantee that the resulting set of $N$-sets is minimal, i.e., no
sursets are present.
Thanks to the ordering of elements in the forest, removing sursets of a given set can be done
efficiently by only checking the sets having a lower prefix.

\paragraph{\textproc{remove}($\G_p,\G,h$)}
Given a set of sets $\G_p$, this function removes all the sets in $\G$ that are sursets of sets in
$\G_p$, starting at level $h$.
If $\G_p$ is $\top$, $\G$ becomes the empty set; otherwise the process is repeated on all the
suffixes having a prefix lower than the prefix of each set to remove.

\paragraph{\textproc{simplify}($\G,h$)}
At level $h$ (initially 1), ranging the prefixes from the higher to the lower, 
each prefixed set is recursively simplified, then any sursets of previously computed sets are
removed from it.
The outputted set of $N$-sets is hence minimal.

\begin{figure}
\hrule\medskip
\begin{small}
\begin{algorithmic}
\Function{remove}{$\G_p,\G,h$}
\If{$\G_p = \top$}
\Return $\bot$
\ElsIf{$\G \in \{\top, \bot\}$}
\Return $\G$
\Else
\For{$j\gets\Call{inds}{\G_p}$}
	\For{$i\gets\Call{inds}{\G}$}
		\If{$i=j$}
		\State $\G(i)\gets \Call{remove}{\G_p(j),\G(i),h+1}$
		\ElsIf{$i<j$}
		\State $\G_{p/j} \gets \{j\mapsto \Call{up}{\G_p(j),h+1}\}$
		\State $\G(i) \gets \Call{remove}{\G_{p/j}, \G(i), h+1}$
		\EndIf
	\EndFor
\EndFor
\State \Return $\G$
\EndIf
\EndFunction
\end{algorithmic}
\end{small}

\medskip
\hrule
\medskip

\begin{small}
\begin{algorithmic}
\Function{simplify}{$\G,h$}
\If{$\G \in \{\top, \bot\}$}
\Return $\G$
\Else
\State $\G'\gets\bot$
\For{$i\gets\Call{inds}{\G}$}
	\State $\G_{/i} \gets \{i\mapsto \Call{simplify}{\G(i),h+1}\}$
	\State $\G_{/i} \gets \Call{remove}{\G', \G_{/i},h}$
	\State $\G'\gets \Call{union}{\G',\G_{/i}}$ 
\EndFor
\State \Return $\G'$
\EndIf
\EndFunction
\end{algorithmic}
\end{small}
\medskip
\hrule
\end{figure}

\end{document}